\documentclass[conference]{IEEEtran}

\usepackage[dvipsnames]{xcolor}
\usepackage{hyperref}
\hypersetup{
    citecolor=RoyalBlue
}
\usepackage[pdftex]{graphicx}
\graphicspath{{./images/}}
\ifCLASSOPTIONcompsoc
    \usepackage[caption=false,font=normalsize,labelfont=sf,textfont=sf]{subfig}
\else
    \usepackage[caption=false,font=footnotesize]{subfig}
\fi

\newif\ifcolor
\colorfalse

\definecolor{theorem}{named}{NavyBlue}
\definecolor{proof}{named}{ForestGreen}
\definecolor{edit}{named}{WildStrawberry}

\usepackage{amsmath, amssymb}
\usepackage[amsmath,amsthm,thmmarks,hyperref,thref]{ntheorem}
\usepackage{physics}
\usepackage{xpatch}

\theoremstyle{plain}
\newtheorem{theorem}{Theorem}
\newtheorem{lemma}{Lemma}
\newtheorem{corollary}{Corollary}

\theoremstyle{definition}
\newtheorem{definition}{Definition}
\newtheorem{hypothesis}{Assumption}

\theoremstyle{remark}

\xpatchcmd{\proof}{\itshape}{\normalfont\proofnamefont}{}{}
\newcommand{\proofnamefont}{\bfseries}

\IEEEoverridecommandlockouts
\usepackage{cite}
\usepackage{array}
\usepackage{enumitem}
\usepackage{xifthen}
\def\BibTeX{{\rm B\kern-.05em{\sc i\kern-.025em b}\kern-.08em
    T\kern-.1667em\lower.7ex\hbox{E}\kern-.125emX}}
\usepackage[normalem]{ulem}

\newif\iflong
\longtrue

\newcommand{\R}{\mathbb{R}}
\newcommand{\C}{\mathcal{C}}

\DeclareMathOperator*{\E}{\mathbf{E}}
\newcommand{\Exp}[2][]
    {\ifthenelse{\isempty{#1}}{\E\left[#2\right]}{\E_{#1}\left[#2\right]}}
\DeclareMathOperator*{\V}{\mathbf{Var}}
\newcommand{\Var}[2][]
    {\ifthenelse{\isempty{#1}}{\V\left[#2\right]}{\V_{#1}\left[#2\right]}}

\renewcommand{\vec}{\mathbf}
\renewcommand{\hat}{\widehat}
\renewcommand{\tilde}{\widetilde}
\renewcommand{\sim}{\thicksim}
\let\originalleft\left
\let\originalright\right
\renewcommand{\left}{\mathopen{}\mathclose\bgroup\originalleft}
\renewcommand{\right}{\aftergroup\egroup\originalright}
\newcommand{\mdots}{,\,\dots\, ,}

\newcommand{\X}{\vec{X}}
\newcommand{\Y}{\vec{Y}}
\newcommand{\Z}{\vec{Z}}
\newcommand{\x}{\vec{x}}
\newcommand{\y}{\vec{y}}
\newcommand{\vt}{\vec{t}}
\newcommand{\z}{\vec{z}}
\newcommand{\xa}{\x_1}
\newcommand{\xb}{\x_2}
\newcommand{\Xa}{\X_1}
\newcommand{\Xb}{\X_2}
\newcommand{\Ya}{\Y_1}
\newcommand{\Yb}{\Y_2}
\newcommand{\fXb}{f_{\Xb}}
\newcommand{\fYa}{f_{\Ya}}
\newcommand{\fZ}{f_{\Z}}
\newcommand{\fz}{f_{\z}}

\newcommand{\boolfunc}[1]{{#1:\{-1,1\}^n\to\{-1,1\}}}
\newcommand{\rboolfunc}[1]{{#1:\{-1,1\}^n\to\R}}

\newcommand{\infl}[2]{\operatorname{\mathbf{Inf}}_{#1}[#2]}

\newcommand{\Ef}{\mathrm{E}_t\fZ}
\newcommand{\Df}{\mathrm{D}_t\fZ}
\newcommand{\Lf}{\mathrm{L}_t\fZ}

\newcommand{\eg}{e.g.,~}
\newcommand{\ie}{i.e.,~}

\title{A Bivariate Invariance Principle\\
\thanks{This work was supported by MIT Welcome Trust Fellowship 2389724,
Lincoln Lab Grant 6944494, MIT Portugal Grant 6942770, and NFS-CNS Grant
6944400.}
}

\author{
\IEEEauthorblockN{
Alexander Mariona\IEEEauthorrefmark{1},
Homa Esfahanizadeh\IEEEauthorrefmark{1},
Rafael G.~L.~D'Oliveira\IEEEauthorrefmark{2}, and
Muriel M\'{e}dard\IEEEauthorrefmark{1}}

\IEEEauthorblockA{
\IEEEauthorrefmark{1}
Research Laboratory of Electronics,
Massachusetts Institute of Technology, Cambridge, MA 02139}

\IEEEauthorblockA{
\IEEEauthorrefmark{2}
School of Mathematical and Statistical Sciences,
Clemson University, Clemson, SC 29634\\
Email: amariona@mit.edu, homaesf@mit.edu, rdolive@clemson.edu, medard@mit.edu}
}

\begin{document}

\maketitle

\begin{abstract}
A notable result from analysis of Boolean functions is the Basic Invariance
Principle (BIP), a quantitative nonlinear generalization of the Central Limit
Theorem for multilinear polynomials. We present a generalization of the BIP for
bivariate multilinear polynomials, \ie polynomials over two
$n$\nobreakdash-length sequences of random variables. This bivariate invariance
principle arises from an iterative application of the BIP to bound the error in
replacing each of the two input sequences. In order to prove this invariance
principle, we first derive a version of the BIP for random multilinear
polynomials, \ie polynomials whose coefficients are random variables.  As a
benchmark, we also state a naive bivariate invariance principle which treats
the two input sequences as one and directly applies the BIP. Neither principle
is universally stronger than the other, but we do show that for a notable class
of bivariate functions, which we term separable functions, our subtler
principle is exponentially tighter than the naive benchmark.
\end{abstract}


\section{Introduction}\label{sec:intro}

Boolean functions are ubiquitous in the fields of complexity theory
\cite{Complexity1, Complexity2}, cryptography \cite{Crypto1, Crypto2}, social
choice theory \cite{Social1, Social2}, and digital electronics \cite{Digital1,
Digital2}.  One particularly significant result from the field of analysis of
Boolean functions is the Basic Invariance Principle (BIP) \cite{BIP}. The BIP
is a nonlinear generalization of the Berry-Esseen Theorem \cite{Berry, Esseen},
which is in turn a quantitative version of the Central Limit Theorem. The
Berry-Esseen Theorem provides an explicit bound on the difference between the
distribution of a finite sum of independent random variables and a standard
Gaussian distribution. The BIP, given a multilinear polynomial and two
differently distributed sequences of independent random variables $\X=(X_1
\mdots X_n)$ and $\Y=(Y_1 \mdots Y_n)$, bounds the expected difference between
$f(\X)$ and $f(\Y)$. This difference can be interpreted as the expected error
incurred by approximating $f(\X)$ as $f(\Y)$.

In order for the bound given by the BIP to be close to 0, the function under
consideration must have relatively low influences. The influence of a
coordinate on a Boolean function quantifies how sensitive the output is to a
change in that particular input coordinate, and the same concept can be
generalized to multilinear polynomials. The notion of influence originated in
social choice and voting theory \cite{Influence}.  Qualitatively, the BIP
states that low-influence functions are invariant to the distribution of the
input sequence.  One notable application of the BIP is to ``replace bits by
Gaussians:'' whether the input is a sequence of uniform random bits or a
sequence of standard Gaussians, the expected output of a low-influence function
does not change too much.

One natural generalization of the BIP would be an invariance principle which
treats functions of two sequences of random variables. Such bivariate functions
open up new options and ideas in applications involving two distinct data
sources, such as in multi-party communication networks,
e.g.,~\cite{MatchedFilter, CWC}. In the context of these models, some functions
are inherently bivariate, even if they could be equivalently written as
univariate functions. In those cases, a bivariate generalization of the BIP may
achieve a tighter bound by exploiting the bivariate structure of the function.

We present one such invariance principle which follows from iteratively
applying the BIP to bound the error in replacing the first input sequence and
then again to bound the error in replacing the second. In order to do so, we
treat the bivariate function as a univariate function with random coefficients
which are determined by the input sequence that is not being replaced at a
given step. To this end, we propose a variation of the BIP which can be applied
to such random functions. For the sake of comparison, we also derive a naive
bivariate invariance principle directly by treating the two input sequences as
a single sequence, effectively viewing the bivariate function as univariate. We
refer to our subtler invariance principle as BVIP-1 and to this naive benchmark
as BVIP-2. Neither principle is universally stronger than the other, but we do
offer one notable example of a family of functions for which BVIP\nobreakdash-1
is exponentially tighter: functions of the form $F(\x,\y)=f(\x)+g(\y)+h(\x\y)$,
which we term \emph{separable} functions. These functions are particularly
interesting because they generalize many different notions of noise that arise
in communication channels, including simple models like the binary symmetric
channel \cite{BSC, CWC}.

The remainder of this paper is organized as follows. In
Section~\ref{sec:prelims}, we summarize concepts and review key results from
analysis of Boolean functions. In Section~\ref{sec:random}, we consider
multilinear random polynomials and propose a version of the BIP for those
random functions in anticipation of Section~\ref{sec:subtle}, in which we
propose and prove BVIP\nobreakdash-1. In Section~\ref{sec:discussion}, we
compare BVIP-1 to the naive benchmark of BVIP-2, present corollary invariance
principles for the special case of separable functions, and offer concluding
thoughts.

\section{Preliminaries}\label{sec:prelims}

We denote random variables with uppercase letters, \eg$X$. We denote vectors
(often referred to as sequences in our context) with bold-faced letters,
\eg$\vec{x}$. Accordingly, vectors of random variables are denoted with
uppercase bold-faced letters, \eg$\vec{X}$. We denote the coordinates (or
elements) of vectors with indexed letters, \eg$x_i$. We sometimes specify the
coordinates of a vector like $\vec{x}=(x_1,x_2 \mdots x_n)$.
Multiplication of two vectors is performed elementwise and results in a new
vector, \ie$\vec{xy} = (x_1y_1,x_2y_2 \mdots x_ny_n)$.  We denote
the set containing the element $i$ with $S\ni i$. We denote the set
$\{1,2 \mdots n\}$ with $[n]$ and its power set with $2^{[n]}$.

\subsection{Results from Analysis of Boolean Functions}

We begin with Boolean functions $\boolfunc{f}$ and
real-valued Boolean functions $\rboolfunc{f}$. All results in the following
sections in fact hold for general multilinear polynomials $f:\R^n\to\R$ (the
domain is $\R^n\times\R^n=\R^{2n}$ in the bivariate case), but because many of
the key tools are defined in the context of Boolean functions, we briefly
discuss those functions here before generalizing. All definitions and results
in this section other than \thref{def:biv} are from \cite{BoolBook}.

\begin{theorem}
Every Boolean function $\rboolfunc{f}$ can be uniquely expressed as an
$n$-variate multilinear polynomial,
\begin{equation*}
    f(\vec{x}) = \sum_{S\subseteq [n]} \hat{f}(S) \prod_{i\in S} x_i.
\end{equation*}
\end{theorem}

This expression is called the \emph{Fourier expansion} of $f$ and is determined
by the \emph{Fourier coefficients} of $f$ on $S$ which are given by the
function $\hat{f}:2^{[n]}\to\R$.  Collectively, the coefficients of $f$ are
referred to as the \emph{Fourier spectrum} of $f$. When we refer to the
\emph{degree} of a Boolean function, we are referring to the degree of its
Fourier expansion.  Since every such expansion is multilinear, the degree $k$
of a Boolean function $f$ (or of any multilinear polynomial $f$) is
\begin{equation*}
    k = \max_{\hat{f}(S)\neq 0} \abs{S}.
\end{equation*}

An important property of a Boolean function is the \emph{influence} of each
coordinate of the input on the output of the function.  The influence of a
coordinate quantifies how likely a particular coordinate is to be
\emph{pivotal}. A coordinate $i$ is pivotal for a particular input $\x$ if
negating $x_i$ negates the output $f(\x)$.

\begin{definition}
The \emph{influence} of a coordinate $i$ on a function $\boolfunc{f}$ is
defined to be the probability that $i$ is pivotal for a random input drawn
uniformly:
\begin{equation*}
    \infl{i}{f} = \Pr_{\X\sim\{-1,1\}^n} \left[f(\X) \neq
        f(\X^{\oplus i})\right],
\end{equation*}
where ${\x^{\oplus i} = (x_1 \mdots x_{i-1},-x_i,x_{i+1} \mdots x_n)}$.
\end{definition}

Informally, if we consider $f$ to be a voting rule in a two-party election, the
influence of the $i$th coordinate can be thought of as the ``influence'' or
``power'' of the $i$th voter.  The influences of a real-valued Boolean function
can be defined in a more analytical fashion, but with a very similar meaning.
Such an approach leads to a relation between the influences and the Fourier
spectrum which we treat as a definition.

\begin{definition}
\label{def:inf}
For $\rboolfunc{f}$ and $i\in[n]$, the influence of coordinate $i$ on $f$ is
\begin{equation*}
    \infl{i}{f} = \sum_{S\ni i}\hat{f}(S)^2.
\end{equation*}
\end{definition}
We also use \thref{def:inf} for general multilinear polynomials, a choice
which is justified by \thref{lm:basis} below.\footnote{\thref{lm:basis} states
that Parseval's theorem holds for multilinear polynomials which are applied to
sequences satisfying \thref{hy:11.65}. It is because of this that we are
justified in using \thref{def:inf}. See \cite[ch.~8.2]{BoolBook} for more
detail.}

We now present a few statements in anticipation of the BIP. First, the BIP only
holds for sequences of random variables with well-behaved distributions. In
particular, we make the following assumption on each random variable in the two
sequences with which we are concerned.\footnote{A slightly different
form of the BIP holds for a looser set of assumptions. \thref{hy:11.65} is
a simpler hypothesis which keeps the bounds tidy and will suffice for our
purposes. See \cite[sec.~3.3]{BIP} for more detail.}

\begin{hypothesis}
\label{hy:11.65}
The random variable $X_i$ satifies $\Exp{X_i}=0$, $\Exp{X_i^2}=1$,
$\Exp{X_i^3}=0$, and $\Exp{X_i^4}\leq 9$.
\end{hypothesis}
Two examples of random variables satisfying \thref{hy:11.65} are a
uniform $\pm 1$ random bit and a standard Gaussian.

The following two lemmas are used to prove the BIP by the replacement method,
and we use them to similar effect in Section~\ref{sec:random}.
\thref{lm:bonami} is a simple hypercontractivity result.

\begin{lemma}[Bonami Lemma]
\label{lm:bonami}
Let $\X=(X_1 \mdots X_n)$ be a sequence of independent but not necessarily
identically distributed random variables satisfying $\Exp{X_i}=\Exp{X_i^3}=0$
and $\Exp{X^4}\leq 9\Exp{X^2}^2$. Let $f$ be a multilinear polynomial of degree
at most $k$. Then
\begin{equation*}
    \Exp{f(\X)^4}\leq 9^k\cdot\Exp{f(\X)^2}^2.
\end{equation*}
\end{lemma}


\begin{lemma}
\label{lm:basis}
Let $f:\R^n\to\R$ be an $n$-variate multilinear polynomial over the
sequence of indeterminates $\x=(x_1 \mdots x_n)$,
\begin{equation*}
    f(\x) = \sum_{S\subseteq[n]}\hat{f}(S)\prod_{i\in S}x_i.
\end{equation*}
When considering a sequence of independent random variables ${\X=(X_1 \mdots
X_n)}$ with ${\Exp{X_i}=0}$ and ${\Exp{X_i^2}=1}$, the parity functions
${\chi_S=\prod_{i\in S} X_i}$ are orthonormal, and hence
\begin{equation*}
    \Exp{f(\X)^2} = \sum_{S\subseteq[n]} \hat{f}(S)^2.
\end{equation*}
\end{lemma}

This leads us to the formal statement of the BIP.

\begin{theorem}[BIP]\label{th:bip}
Let $f:\R^n\to\R$ be an $n$-variate multilinear polynomial of degree at most
$k$. Let $\X$ and $\Y$ be $n$\nobreakdash-length sequences of independent
random variables satisfying \thref{hy:11.65}. Let $\psi:\R\to\R$ be
$\mathcal{C}^4$, \ie the derivatives $\psi' \mdots \psi''''$ exist and are
continuous, with $\norm{\psi''''}_\infty\leq C$. Then
\begin{equation*}
    \abs{\Exp{\psi(f(\X))} - \Exp{\psi(f(\Y))}} \leq
    \frac{C}{12}\cdot 9^k \cdot \sum_{t=1}^n \infl{t}{f}^2.
\end{equation*}
\end{theorem}

The function $\psi$ used in the BIP is called a \emph{test function} or a
\emph{distinguisher}, and is used to specify a particular notion of
``closeness'' between two random variables. A natural measure is
\emph{cdf-closeness}, which is used in the Berry-Esseen Theorem. Two random
variables $X$ and $Y$ are cdf-close if ${\Pr{X\leq u}\approx\Pr{Y\leq u}}$ for
all $u\in\R$. Equivalently, two random variables are cdf-close if
$\Exp{\psi(X)}\approx\Exp{\psi(Y)}$ with $\psi(s)=1_{s\leq u}$ for all
$u\in\R$.  The BIP is powerful enough to give bounds on cdf-closeness and many
other notions of closeness.\footnote{Note that continuity of $\psi''''$ is
required.  Smoothing techniques can be used to approximate functions like
${\psi(s)=1_{s\leq u}}$.  There is of course a tradeoff between the quality of
the approximation and the magnitude of the fourth derivative of the smoothed
function. See \cite[ch.~11]{BoolBook} for more detail.}

\subsection{Bivariate Functions}

Finally, we specify the class of functions which we will refer to throughout
this paper simply as bivariate functions.

\begin{definition}
\label{def:biv}
An $n$-bivariate multilinear polynomial function ${f:\R^n\times\R^n\to\R}$ over
the sequences of indeterminates ${\xa=(x_{1,1} \mdots x_{1,n})}$ and
${\xb=(x_{2,1} \mdots x_{2,n})}$ is a function of the form
\begin{equation*}
    f(\xa,\xb) = \sum_{S_1,S_2\subseteq[n]} \hat{f}(S_1,S_2)
    \prod_{i\in S_1}x_{1,i}\prod_{j\in S_2}x_{2,j}.
\end{equation*}
\end{definition}

The form given in \thref{def:biv} suggests that $\hat{f}(S_1,S_2)$ is Fourier
coefficient. Indeed, if we consider $f$ to instead be a function of the
concatenated sequence $\x=\xa\|\xb$, then $\hat{f}(S_1,S_2)$ is the Fourier
coefficient on the set $S_1\cup S_2^+$, where ${S_2^+ = \{i+n: i\in S_2\}}$.
Nonetheless, we will not consider any subtleties of Fourier theory for
bivariate functions and we do not make any claims about any of the classic
Fourier identities in this bivariate basis.

\section{Random Functions}\label{sec:random}

In anticipation of BVIP-1, we introduce in this section the concept of random
multilinear polynomials and prove a version of the BIP for these functions.

\begin{definition}
\label{def:random}
A random $n$-variate multilinear polynomial ${\fZ:\R^n\to\R}$ is a multilinear
polynomial over the sequence of indeterminates $\x=(x_1 \mdots x_n)$ whose
coefficients ${\hat{\fZ}(S)\in\R}$ are random variables which are wholly
determined by the random variable $\Z$:
\begin{equation*}
    \fZ(\x) = \sum_{S\subseteq[n]} \hat{\fZ}(S)\prod_{i\in S}x_i.
\end{equation*}
The influence of coordinate $i$ on $\fZ$ is a random variable which is defined
to be
\begin{equation*}
    \infl{i}{\fZ} = \sum_{S\ni i}\hat{\fZ}(S)^2.
\end{equation*}
\end{definition}
We think of $\Z$ as the random variable which controls $\fZ$ or, alternatively,
which describes the randomness of $\fZ$. We will sometimes refer to random
multilinear polynomials simply as random functions. Such functions will always
be univariate.

In pursuit of a version of the BIP for random functions, we start with
appropriate corollaries of \thref{lm:bonami} and \thref{lm:basis}.

\begin{corollary}
\label{cr:rbonami}
Let $\X=(X_1 \mdots X_n)$ be a sequence of independent but not necessarily
identically distributed random variables satisfying the requirement that
$\Exp{X_i}=\Exp{X_i^3}=0$ and ${\Exp{X_i^4}\leq 9 \Exp{X_i^2}^2}$. Let $\fZ$ be
a random multilinear polynomial of degree at most $k$. Then
\begin{equation*}
    \Exp[\X,\Z]{\fZ(\X)^4} \leq 9^k\cdot \Exp[\X,\Z]{\fZ(\X)^2}^2.
\end{equation*}
\end{corollary}

\iflong
\begin{proof}
We expand the expectation over $\Z$ using the law of total expectation. Without
loss of generality, assume that $\Z$ is a random variable over a discrete
sample space $\mathcal{Z}$. Then we can write
\begin{align*}
    \Exp[\X,\Z]{\fZ(\X)^4} &= \sum_{\z\in\mathcal{Z}}
        \Pr{\Z=\z}\Exp[\X]{\fZ(\X)^4 \mid \Z=\z} \\
    &= \sum_{\z\in\mathcal{Z}} \Pr{\Z=\z}\Exp[\X]{\fz(\X)^4},
\end{align*}
where $\fz$ is the function $\fZ$ given that $\Z=\z$. Conditioning on $\Z=\z$
fixes the coefficients of $\fZ$, allowing us to apply \thref{lm:bonami}
directly.
\begin{align*}
    \Exp[\X,\Z]{\fZ(\X)^4} &\leq \sum_{\z\in\mathcal{Z}}
        \Pr{\Z=\z}\cdot 9^k\cdot \Exp[\X]{\fz(\X)^2}^2 \\
    &= 9^k \cdot \Exp[\X,\Z]{\fZ(\X)^2}^2.
\end{align*}
If $\mathcal{Z}$ is a continuous space, then we must integrate
over $\mathcal{Z}$ and consider the pdf of $\Z$, but the argument is otherwise
identical.
\end{proof}
\fi

\begin{corollary}
\label{cr:rbasis}
Let $\fZ$ be a random $n$-variate multilinear polynomial. When considering a
sequence of independent random variables $\X=(X_1 \mdots X_n)$ satisfying
$\Exp{X_i}=0$ and ${\Exp{X_i^2}=1}$, the parity functions $\chi_S = \prod_{i\in
S} X_i$ are orthonormal, and hence
\begin{equation*}
    \Exp[\X,\Z]{\fZ(\X)^2} = \Exp[\Z]{\sum_{S\subseteq [n]} \hat{\fZ}(S)^2}.
\end{equation*}
\end{corollary}

\iflong
\begin{proof}
As in the proof of \thref{cr:rbonami}, we expand the expectation over $\Z$
using the law of total expectation.
\begin{align*}
    \Exp[\X,\Z]{\fZ(\X)^2} &= \sum_{\z\in\mathcal{Z}} \Pr{\Z=\z}
        \Exp[\X]{\fZ(\X)^2 \mid \Z=\z} \\
    &= \sum_{\z\in\mathcal{Z}} \Pr{\Z=\z}
        \Exp[\X]{\fz(\X)^2}.
\end{align*}
Conditioning on $\Z=\z$, we apply \thref{lm:basis} directly.
\begin{align*}
    \Exp[\X,\Z]{\fZ(\X)^2} &= \sum_{\z\in\mathcal{Z}} \left( \Pr{\Z=\z}
        \sum_{S\subseteq [n]}\hat{\fz}(S)^2\right) \\
    &= \Exp[\Z]{\sum_{S\subseteq [n]} \hat{\fZ}(S)^2}.
\end{align*}
Again, if $\mathcal{Z}$ is a continuous space, we must instead integrate over
$\mathcal{Z}$ and consider the pdf of $\Z$ to the same effect.
\end{proof}
\fi

Our BIP for random functions is identical in spirit to \thref{th:bip}, but the
resulting upper bound is in terms of the expected influences of the given
function.

\begin{theorem}[BIP for Random Functions]
\label{th:rbip}
Let $\fZ$ be a random $n$-variate multilinear polynomial of degree at most $k$.
Let $\X$ and $\Y$ be $n$\nobreakdash-length sequences of independent random
variables satisfying \thref{hy:11.65}. Let $\psi:\R\to\R$ be $\mathcal{C}^4$.
Then
\begin{equation*}
    \abs{\Exp[\X,\Z]{\psi(\fZ(\X))} - \Exp[\Y,\Z]{\psi(\fZ(\Y))}} \leq
    \frac{C}{12}\cdot 9^k \cdot \sum_{t=1}^n \Exp[\Z]{\infl{t}{\fZ}}^2.
\end{equation*}
\end{theorem}

\iflong
\begin{proof}
The proof closely follows the proof of the BIP given in
\cite[ch.~11.6]{BoolBook}, so we omit some exposition which can be found there.
Nonetheless, for completeness we summarize the arguments and highlight the
points where the random functions affect the process.

We use the replacement method and define
\begin{equation*}
    H_t = \fZ(Y_1 \mdots Y_t, X_{t+1} \mdots X_n),
\end{equation*}
such that $H_0=\fZ(\X)$ and $H_n=\fZ(\Y)$. We show that
\begin{equation}
\label{eq:wantstaylor}
    \abs{\Exp[\X,\Y,\Z]{\psi(H_{t-1})-\psi(H_t)}} \leq \frac{C}{12}\cdot 9^k
    \cdot \Exp[\Z]{\infl{t}{\fZ}}^2.
\end{equation}
Summing over $t$ and applying the triangle inequality will complete the proof.

Let the random functions $\Ef$ and $\Df$ be defined as
\begin{align*}
    \Ef(x) &= \sum_{S\not\ni t}\hat{\fZ}(S)\prod_{i\in S}x_i \\
    \Df(x) &= \sum_{S\ni t}\hat{\fZ}(S)\prod_{i\in S\setminus \{t\}}x_i,
\end{align*}
such that $\fZ(\x)=\Ef(\x)+x_t \Df(\x)$. Since neither $\Ef$ nor $\Df$ depends
on $x_t$, we can define
\begin{align*}
    U_t &= \Ef(Y_1,\dots,Y_{t-1},\cdot,X_{t+1},\dots,X_{n}), \\
    \Delta_t &= \Df(Y_1,\dots,Y_{t-1},\cdot,X_{t+1},\dots,X_{n}),
\end{align*}
so that
\begin{equation*}
    H_{t-1} = U_t + \Delta_t X_t, \quad H_t = U_t + \Delta_t Y_t.
\end{equation*}

We can then bound \eqref{eq:wantstaylor} by taking 3rd-order Taylor expansions
of $\psi(H_{t-1})$ and $\psi(H_t)$ and then taking the difference between them.
After subtracting and taking expectations over $\X$, $\Y$, and $\Z$, the
0th-order terms cancel directly, and the 1st-, 2nd-, and 3rd-order terms cancel
because $X_t$ and $Y_t$ are independent of $U_t$ and $\Delta_t$ and $X_t$ and
$Y_t$ have matching 1st, 2nd, and 3rd moments. For the 4th-order error term, we
apply the triangle inequality and make use of the assumption that
$\abs{\psi''''(U_t^{\ast})},\abs{\psi''''(U_t^{\ast\ast})}\leq C$ to upper
bound the left-hand side of \eqref{eq:wantstaylor} by
\begin{equation*}
    \frac{C}{24}\cdot\left( \Exp[\X,\Y,\Z]{(\Delta_t X_t)^4} +
    \Exp[\X,\Y,\Z]{(\Delta_t Y_t)^4}\right).
\end{equation*}

All that remains is to bound
\begin{equation*}
    \Exp[\X,\Y,\Z]{(\Delta_t X_t)^4}, \Exp[\X,\Y,\Z]{(\Delta_t Y_t)^4}
    \leq 9^k \cdot \Exp[\Z]{\infl{t}{\fZ}}^2,
\end{equation*}
which can be done using \thref{cr:rbonami}. We give details for the case of
$\Exp[\X,\Y,\Z]{(\Delta_t X_t)^4}$. The case for $\Exp[\X,\Y,\Z]{(\Delta_t
Y_t)^4}$ is identical. Define
\begin{equation*}
    \Lf(\x) = x_t \Df(\x) = \sum_{S\ni t}\hat{f}(S)\prod_{i\in S}x_i.
\end{equation*}
Then, $\Delta_t X_t=\Lf(Y_1 \mdots X_t \mdots X_n)$. Since $\Lf$ has degree
at most $k$ we can apply \thref{cr:rbonami} to obtain
\begin{equation}
\label{eq:crbound}
    \Exp[\X,\Y,\Z]{\left(\Delta_t X_t\right)^4} \leq 9^k \cdot
    \Exp[\X,\Y,\Z]{(\Delta_t X_t)^2}^2.
\end{equation}
Finally, since the elements of $\X$ and $\Y$ all have mean 0 and 2nd moment 1,
by \thref{cr:rbasis}
\begin{align}
\nonumber
    \Exp[\X,\Y,\Z]{\left(\Delta_t X_t\right)^2} &=
    \Exp[\Z]{\sum_{S\subseteq[n]}\hat{\Lf}(S)^2} \\
\nonumber
    &= \Exp[\Z]{\sum_{S\ni t}\hat{\fZ}(S)^2} \\
\label{eq:infbound}
    &= \Exp[\Z]{\infl{t}{\fZ}}.
\end{align}
Combining \eqref{eq:crbound} and \eqref{eq:infbound}, we have that
\begin{equation*}
    \Exp[\X,\Y,\Z]{(\Delta_t X_t)^4} \leq 9^k \cdot \Exp[\Z]{\infl{t}{\fZ}}^2,
\end{equation*}
which completes the proof.
\end{proof}
\fi

\section{A Bivariate Invariance Principle}\label{sec:subtle}

We now present BVIP-1. To derive it, we iteratively apply the BIP
to replace each input sequence in turn. The first step in this process is to
treat the input sequence which is not currently being replaced as a random
parameter of the function, allowing us to view the bivariate function as a
random univariate function. We can then use the BIP for random functions to
bound the error incurred by this replacement.

\begin{theorem}[BVIP-1]
\label{th:subtle}
Let $f$ be an $n$-bivariate multilinear polynomial in which each term includes
at most $k$ elements from each input sequence:
\begin{equation*}
    f(\xa,\xb) = \sum_{S_1,S_2\subseteq [n]}
        \hat{f}(S_1,S_2) \prod_{i\in S_1}x_{1,i} \prod_{j\in S_2}x_{2,j},
\end{equation*}
where $\hat{f}(S_1,S_2)=0$ if $\abs{S_1}>k$ or $\abs{S_2}>k$.  Let
$\Xa$, $\Xb$, $\Ya$, and $\Yb$ be $n$-length sequences of independent
random variables satisfying \thref{hy:11.65}. Assume $\psi:\R\to\R$ is
$\mathcal{C}^4$ with $\norm{\psi''''}_\infty\leq C$. Then
\begin{equation}
\label{eq:bvip-1}
    \abs{E_X - E_Y} \leq \frac{C}{12}\cdot 9^{k} \cdot \sum_{t=1}^n
        \left(\tilde{\Sigma}_{1,t}^2 + \tilde{\Sigma}_{2,t}^2\right),
\end{equation}
where
\begin{align*}
    E_X &= \Exp[\Xa,\Xb]{\psi(f(\Xa,\Xb))} \\
    E_Y &= \Exp[\Ya,\Yb]{\psi(f(\Ya,\Yb))} \\
    \tilde{\Sigma}_{1,t} &= \sum_{S_1\ni t}\abs{T_2(S_1)}
        \sum_{S_2\in T_2(S_1)}\hat{f}(S_1,S_2)^2 \\
    \tilde{\Sigma}_{2,t} &= \sum_{S_2\ni t}\abs{T_1(S_2)}
        \sum_{S_1\in T_1(S_2)}\hat{f}(S_1,S_2)^2,
\end{align*}
and $T_2(S_1)$ and $T_1(S_2)$ are the sets
\begin{align*}
    T_2(S_1) &= \left\{ S_2\subseteq[n]:\abs{S_2}\leq k,\,
        \hat{f}(S_1,S_2)\neq0 \right\} \\
    T_1(S_2) &= \left\{S_1\subseteq[n]:\abs{S_1}\leq k,\,
        \hat{f}(S_1,S_2)\neq0 \right\}.
\end{align*}
\end{theorem}

\begin{proof}
As described above, our strategy is to define random functions $\fXb$ and
$\fYa$ such that $\fXb(\Xa) = f(\Xa,\Xb)$ and $\fYa(\Xb) = f(\Ya,\Xb)$.
Applying \thref{th:rbip} to $\fXb$ bounds the error incurred by replacing
$\Xa$ with $\Ya$. An application of the same theorem to $\fYa$ bounds the error
in replacing $\Xb$ and $\Yb$. Computing the expected influences of the random
functions and using the triangle inequality will complete the proof.

We begin by constructing the desired random functions. Let $\fXb:\R^n\to\R$ be
defined as
\begin{align*}
    \fXb(\vt) &= f(\vt,\Xb) \\
    &= \sum_{S_1\subseteq [n]} \left[\sum_{S_2\subseteq [n]}
    \hat{f}(S_1,S_2) \prod_{j \in S_2} X_{2,j}\right]
    \prod_{i\in S_1}t_i \\
    &= \sum_{S_1\subseteq [n]} \hat{\fXb}(S_1)\prod_{i\in S_1}t_i.
\end{align*}
Similarly, let $\fYa:\R^n\to\R$ be defined as
\begin{align*}
    \fYa(\vt) &= f(\Ya,\vt) \\
    &= \sum_{S_2\subseteq [n]} \left[\sum_{S_1\subseteq [n]}
    \hat{f}(S_1,S_2) \prod_{i \in S_1} Y_{1,j}\right]
    \prod_{i\in S_2}t_i \\
    &= \sum_{S_2\subseteq [n]} \hat{\fYa}(S_2)\prod_{i\in S_2}t_i.
\end{align*}
Note that both $\fXb$ and $\fYa$ are of degree at most $k$ and that
$\fXb(\Ya) = \fYa(\Xb)$. From the definitions of $E_X$ and $E_Y$,
\begin{align*}
    E_X &= \Exp[\Xa,\Xb]{\psi(\fXb(\Xa))} \\
    E_Y &= \Exp[\Ya,\Yb]{\psi(\fYa(\Yb))}.
\end{align*}
By analogy, let
\begin{align*}
    E_{XY} &= \Exp[\Ya,\Xb]{\psi(f(\Ya,\Xb)} \\
    &= \Exp[\Ya,\Xb]{\psi(\fXb(\Ya)} \\
    &= \Exp[\Ya,\Xb]{\psi(\fYa(\Xb)}.
\end{align*}

We upper bound the quantity of interest as
\begin{equation}
\label{eq:subt-tri}
    \abs{E_X - E_Y} \leq \abs{E_X - E_{XY}} + \abs{E_{XY} - E_Y}.
\end{equation}
Applying \thref{th:rbip} to each term on the right-hand side of
\eqref{eq:subt-tri} yields
\begin{align}
\label{eq:deste1}
\abs{E_X-E_{XY}} &\leq \frac{C}{12}\cdot 9^k \cdot \sum_{t=1}^n
    \Exp[\Xb]{\infl{t}{\fXb}}^2 \\
\label{eq:deste2}
\abs{E_{XY}-E_Y} &\leq \frac{C}{12}\cdot 9^k \cdot \sum_{t=1}^n
        \Exp[\Ya]{\infl{t}{\fYa}}^2.
\end{align}

All that remains is to bound the expected influences of $\fXb$ and $\fYa$. We
handle the case of $\fXb$ explicitly, with the argument for $\fYa$ being
identical. For convenience, define
\begin{equation*}
    \sigma_2(S_1) = \sum_{S_2\in T_2(S_1)}\hat{f}(S_1,S_2)^2.
\end{equation*}
We have
\begin{align}
\nonumber
    &\Exp[\Xb]{\infl{t}{\fXb}} = \Exp[\Xb]{\sum_{S_1\ni t}\hat{\fXb}(S_1)^2} \\
\nonumber
    &\quad= \Exp[\Xb]{\sum_{S_1\ni t} \left(\sum_{S_2\subseteq [n]}
        \hat{f}(S_1,S_2) \prod_{j \in S_2} X_{2,j}\right)^2} \\
\label{eq:set}
    &\quad= \sum_{S_1\ni t} \Exp[\Xb]{\left(\sum_{S_2\in T_2(S_1)}
        \hat{f}(S_1,S_2) \prod_{j \in S_2} X_{2,j}\right)^2} \\
\label{eq:csineq}
    &\quad\leq \sum_{S_1\ni t} \Exp[\Xb]{\sigma_2(S_1)
        \left(\sum_{S_2\in T_2(S_1)} \prod_{j \in S_2} X_{2,j}^2\right)} \\
\label{eq:linexp}
    &\quad= \sum_{S_1\ni t} \sigma_2(S_1) \left(\sum_{S_2\in T_2(S_1)}
        \Exp[\Xb]{\prod_{j \in S_2} X_{2,j}^2}\right) \\
\label{eq:x2indp}
    &\quad= \sum_{S_1\ni t} \sigma_2(S_1) \left(\sum_{S_2\in T_2(S_1)}
    \prod_{j \in S_2} \Exp[\Xb]{X_{2,j}^2}\right) \\
\label{eq:x2mom2}
    &\quad= \sum_{S_1\ni t} \abs{T_2(S_1)} \cdot \sigma_2(S_1) \\
\label{eq:sigma1}
    &\quad= \tilde{\Sigma}_{1,t},
\end{align}
where \eqref{eq:set} follows from linearity of expectation and the fact that
for a given $S_1$, we have $\hat{f}(S_1,S_2)\neq0$ only if ${S_2\in T_2(S_2)}$;
\eqref{eq:csineq} follows from the Cauchy-Schwarz inequality; \eqref{eq:linexp}
again follows from linearity of expectation; \eqref{eq:x2indp} follows from the
assumption that the elements of $\Xb$ are independent; and \eqref{eq:x2mom2}
follows from the assumption that $\Exp{X_{2,j}^2}=1$ for all $j\in[n]$. The
same argument applied to $\fYa$ gives
\begin{equation}
\label{eq:sigma2}
    \Exp[\Ya]{\infl{t}{\fYa}} \leq \tilde{\Sigma}_{2,t}.
\end{equation}

Substituting \eqref{eq:sigma1} and \eqref{eq:sigma2} into \eqref{eq:deste1} and
\eqref{eq:deste2} respectively yields
\begin{align*}
    \abs{E_X-E_{XY}} &\leq \frac{C}{12}\cdot 9^k \cdot \sum_{t=1}^n
        \tilde{\Sigma}_{1,t}^2 \\
    \abs{E_{XY}-E_Y} &\leq \frac{C}{12}\cdot 9^k \cdot \sum_{t=1}^n
        \tilde{\Sigma}_{2,t}^2,
\end{align*}
which, combined with \eqref{eq:subt-tri}, completes the proof.
\end{proof}

\section{Discussion and Conclusion}\label{sec:discussion}

As a baseline against which we can compare the bound of \thref{th:subtle}, we
state as a corollary the bound which the BIP yields when we treat a bivariate
function as a univariate function. Such a univariate interpretation takes as
input the concatenation of the two sequences which are the inputs to the
original bivariate function. We refer to this naive bivariate invariance
principle as BVIP-2.

\begin{corollary}[BVIP-2]
\label{cr:naive}
Let $f$ be an $n$-bivariate multilinear polynomial in which each term includes
at most $k$ elements from each input sequence:
\begin{equation*}
    f(\xa,\xb) = \sum_{S_1,S_2\subseteq [n]}
        \hat{f}(S_1,S_2) \prod_{i\in S_1}x_{1,i} \prod_{j\in S_2}x_{2,j},
\end{equation*}
where $\hat{f}(S_1,S_2)=0$ if $\abs{S_1}>k$ or $\abs{S_2}>k$.  Let
$\Xa$, $\Xb$, $\Ya$, and $\Yb$ be $n$-length sequences of independent
random variables satisfying \thref{hy:11.65}. Assume $\psi:\R\to\R$ is
$\mathcal{C}^4$ with $\norm{\psi''''}_\infty\leq C$. Then
\begin{equation}
\label{eq:bvip-2}
    \abs{E_X - E_Y} \leq \frac{C}{12}\cdot 9^{2k} \cdot \sum_{t=1}^n
        \left(\Sigma_{1,t}^2 + \Sigma_{2,t}^2\right),
\end{equation}
where $E_X$, $E_Y$, $T_2(S_1)$ and $T_1(S_2)$ are as defined in
\thref{th:subtle}, and
\begin{align*}
    \Sigma_{1,t} &= \sum_{S_1\ni t}\sum_{S_2\in T_2(S_1)}\hat{f}(S_1,S_2)^2 \\
    \Sigma_{2,t} &= \sum_{S_2\ni t}\sum_{S_1\in T_1(S_2)}\hat{f}(S_1,S_2)^2.
\end{align*}
\end{corollary}

\iflong
\begin{proof}
The strategy is to define a univariate function $g$ which is equivalent to $f$
when the two input sequences are considered as a single sequence so that we may
then apply the BIP to $g$. Given a particular subset $S\subseteq[2n]$, let
${S^*_1=S\cap [n]}$, ${\tilde{S}^*_2=S\cap\{n+1 \mdots 2n\}}$, and
${S^*_2=\{i:n+i\in \tilde{S}^*_2\}}$. Then let $g$ be a $2n$-variate
multilinear polynomial of degree such that
\begin{equation*}
    g(\x) = \sum_{S\subseteq[2n]}
        \hat{f}(S^*_1,S^*_2)\prod_{i\in S}x_i.
\end{equation*}
For any $n$-length sequences $\xa$ and $\xb$, it is clear that
$g(\xa\|\xb)=f(\xa,\xb)$, where $\xa\|\xb$ is the concatentation of $\xa$ and
$\xb$. Furthermore, since $\hat{f}(S_1,S_2)=0$ when $\abs{S_1}>k$ or
$\abs{S_2}>k$, by construction $g$ is of degree at most $2k$.

Applying the BIP to $g$ for the concatenations $\X=\Xa\|\Xb$ and $\Y=\Ya\|\Yb$
yields
\begin{equation}
\label{eq:naivesub}
    \abs{\Exp{\psi(g(\X))} - \Exp{\psi(g(\Y))}} \leq \frac{C}{12} \cdot
    9^{2k}\cdot \sum_{t=1}^{2n}\infl{t}{g}^2.
\end{equation}
We now compute $\infl{t}{g}$ in terms of the coefficients of $f$. By
definition, $\hat{g}(S)=\hat{f}(S^*_1,S^*_2)$. Thus, for $t\in[n]$,
\begin{align}
\nonumber
    \infl{t}{g} &= \sum_{S\ni t}\hat{G}(S)^2 \\
\nonumber
    &= \sum_{S_1\ni t}\sum_{S_2\subseteq[n]}\hat{f}(S_1,S_2)^2 \\
\label{eq:t2s1}
    &= \sum_{S_1\ni t}\sum_{S_2\in T_2(S_1)}\hat{f}(S_1,S_2)^2 \\
\nonumber
    &= \Sigma_{t,1}.
\end{align}
where \eqref{eq:t2s1} follows from the definition of the set $T_2(S_1)$.
By a parallel argument, for $t\in\{n+1 \mdots n\}$,
\begin{align}
\nonumber
    \infl{t}{g} &= \sum_{S\ni t}\hat{G}(S)^2 \\
\nonumber
    &= \sum_{S_2\ni t-n}\sum_{S_1\subseteq[n]}\hat{f}(S_1,S_2)^2 \\
\label{eq:t1s2}
    &= \sum_{S_2\ni t-n}\sum_{S_1\in T_1(S_2)}\hat{f}(S_1,S_2)^2 \\
\nonumber
    &= \Sigma_{t-n,2}.
\end{align}

Combining \eqref{eq:t2s1} and \eqref{eq:t1s2},
\begin{align}
\nonumber
    \sum_{t=1}^{2n} \infl{t}{g}^2 &= \sum_{t=1}^n \infl{t}{g}^2 +
        \sum_{t=n+1}^{2n} \infl{t}{g}^2 \\
\nonumber
    &= \sum_{t=1}^n \Sigma_{t,1}^2 + \sum_{t=n+1}^{2n}\Sigma_{t-n,2}^2 \\
\label{eq:naivepenult}
    &= \sum_{t=1}^n \Sigma_{t,1}^2 + \Sigma_{t,2}^2.
\end{align}
Substituting \eqref{eq:naivepenult} into \eqref{eq:naivesub} yields the desired
inequality after replacing $g(\X)$ and $g(\Y)$ on the left-hand side of
\eqref{eq:naivesub} with $f(\Xa,\Xb)$ and $f(\Ya,\Yb)$.
\end{proof}
\fi

Note that in the context of BVIP-2, $k$ is not strictly speaking
the degree of $f$, as it is in the BIP. Indeed, the degree of $f$ can be as
large as $2k$ here, and as such, the bound incurs a factor of $9^{2k}$ directly
from the BIP.

Comparing the bounds of BVIP-1 and BVIP-2 given in \eqref{eq:bvip-1} and
\eqref{eq:bvip-2} respectively, the main differences are a factor of $9^{k}$
versus $9^{2k}$ and the quantity $\tilde{\Sigma}_{i,t}$ versus $\Sigma_{i,t}$
(for $i\in\{1,2\}$), which we recall are defined (for $i=1$) as
\begin{align*}
    \tilde{\Sigma}_{1,t} &= \sum_{S_1\ni t}\abs{T_2(S_1)}
        \sum_{S_2\in T_2(S_1)}\hat{f}(S_1,S_2)^2 \\
    \Sigma_{1,t} &= \sum_{S_1\ni t}\sum_{S_2\in T_2(S_1)}\hat{f}(S_1,S_2)^2.
\end{align*}
Thus, BVIP-1 trades a factor of $9^k$ compared to BVIP-2 in exchange for
counting $\abs{T_2(S_1)}$ for each $S_1$ (and likewise $\abs{(T_1(S_2)}$ for
each $S_2$). We can conceptualize $\abs{T_2(S_1)}$ and $\abs{T_1(S_2)}$ as
measuring the ``strength'' of the interaction in $f$ between the inputs $\Xa$
and $\Xb$. If those cardinalities are large, then there are many terms of $f$
in which some coordinates of $\Xb$ are multiplied with the coordinates of
$\Xa$. Note that $\abs{(T_1(S_2)}$ and $\abs{T_2(S_1)}$ both arise from
applying the Cauchy-Schwarz inequality, as in \eqref{eq:csineq}, and are hence
upper bounds on this interaction strength.

The question of whether BVIP-1 outperforms BVIP-2 for a particular $f$ is thus
a question of whether the interaction between $\Xa$ and $\Xb$ is ``small
enough'' to beat the extra factor of $9^k$ incurred by BVIP-2. As a concrete
example of a family of functions for which BVIP\nobreakdash-1 is always tighter
than BVIP-2, consider \emph{separable} bivariate functions.


\begin{definition}
An $n$-bivariate multilinear polynomial function ${F:\R^n\times\R^n\to\R}$ is
separable into $f$, $g$, and $h$ if it can be written in terms of $n$-variate
multilinear polynomials $f,g,h:\R^n\to\R$ like
\begin{equation*}
    F(\xa,\xb) = f(\xa) + g(\xb) + h(\xa\xb).
\end{equation*}
\end{definition}

For separable functions, the bounds of both bivariate invariance principles can
be cleanly expressed in terms of the influences of $f$, $g$ and $h$, resulting
in a form which is very close to that of the BIP.

\begin{corollary}[Separable BVIP-1]
\label{cr:sepsubt}
Let $F$ be an $n$-bivariate multilinear polynomial which is separable into $f$,
$g$, and $h$, each of which is of degree at most $k$. Let $\Xa$, $\Xb$, $\Ya$,
and $\Yb$ be $n$-length sequences of independent random variables satisfying
\thref{hy:11.65}. Assume $\psi:\R\to\R$ is $\C^4$ with $\norm{\psi''''}\leq C$.
Then
\begin{equation*}
    \abs{E_X-E_Y} \leq \frac{2C}{3}\cdot 9^k \cdot \sum_{t=1}^n
    \left(\infl{t}{f}^2 + \infl{t}{g}^2 + 2\infl{t}{h}^2\right),
\end{equation*}
where $E_X$ and $E_Y$ are as defined in \thref{th:subtle}.
\end{corollary}

\iflong
\begin{proof}
By assumption, $F$ is of the form
\begin{equation*}
    F(\xa,\xb) = f(\xa) + g(\xb) + h(\xa\xb).
\end{equation*}
Since $f$, $g$, and $h$ are each of degree at most $k$, each term of $F$
includes at most $k$ elements from each input sequence. Thus, we can apply
\thref{th:subtle}.

For separable functions, we can compute $T_2(S_1)$ and $T_1(S_2)$ directly. We
have
\begin{equation*}
    T_2(S_1) = \left\{\emptyset, S_1\right\}, \quad
    T_1(S_2) = \left\{\emptyset, S_2\right\},
\end{equation*}
Hence, for a given $S\subseteq[n]$, the only (possibly) non-zero coefficients
of $F$ are
\begin{equation*}
    \hat{F}(S,\emptyset) = \hat{f}(S), \quad
    \hat{F}(\emptyset,S) = \hat{g}(S), \quad
    \hat{F}(S,S) = \hat{h}(S).
\end{equation*}
Computing $\tilde{\Sigma}_{1,t}$, we have
\begin{align*}
    \tilde{\Sigma}_{1,t} &= \sum_{S_1\ni t} \abs{T_2(S_1)}
        \sum_{S_2\ni T_2(S_1)} \hat{F}(S_1,S_2)^2 \\
    &= \sum_{S_1\ni t} 2 \left(
        \hat{F}(S_1,\emptyset)^2 + \hat{F}(S_1,S_1)^2\right) \\
    &= \sum_{S_1\ni t} 2 \left(\hat{f}(S_1)^2 + \hat{h}(S_1)^2\right) \\
    &= 2\sum_{S_1\ni t}\hat{f}(S_1)^2 + 2\sum_{S_1\ni t}\hat{h}(S_1)^2 \\
    &= 2\infl{t}{f} + 2\infl{t}{h}.
\end{align*}
Similarly, we have
\begin{equation*}
    \tilde{\Sigma}_{2,t} = 2\infl{t}{g} + 2\infl{t}{h}.
\end{equation*}
A simple application of Cauchy-Schwarz yields
\begin{equation}
\label{eq:crsubtfin}
    \tilde{\Sigma}_{1,t}^2 + \tilde{\Sigma}_{2,t}^2 \leq
        8\infl{t}{f}^2 + 8\infl{t}{g}^2 + 16\infl{t}{h}^2.
\end{equation}
Substituting \eqref{eq:crsubtfin} into the bound of \thref{th:subtle} yields
the desired inequality.
\end{proof}
\fi

\begin{corollary}[Separable BVIP-2]
Let $F$ be an $n$-bivariate multilinear polynomial which is separable into $f$,
$g$, and $h$, each of which is of degree at most $k$. Let $\Xa$, $\Xb$, $\Ya$,
and $\Yb$ be $n$-length sequences of independent random variables satisfying
\thref{hy:11.65}. Assume $\psi:\R\to\R$ is $\C^4$ with $\norm{\psi''''}\leq C$.
Then
\begin{equation*}
    \abs{E_X-E_Y} \leq \frac{C}{6}\cdot 9^{2k} \cdot \sum_{t=1}^n
    \left(\infl{t}{f}^2 + \infl{t}{g}^2 + 2\infl{t}{h}^2\right),
\end{equation*}
where $E_X$ and $E_Y$ are as defined in \thref{th:subtle}.
\end{corollary}

\iflong
\begin{proof}
As in the proof of \thref{cr:sepsubt}, we again have
\begin{equation*}
    T_2(S_1) = \left\{\emptyset, S_1\right\}, \quad
    T_1(S_2) = \left\{\emptyset, S_2\right\},
\end{equation*}
with the possibly non-zero coefficients for a given $S\subseteq[n]$ being
\begin{equation*}
    \hat{F}(S,\emptyset) = \hat{f}(S), \quad
    \hat{F}(\emptyset,S) = \hat{g}(S), \quad
    \hat{F}(S,S) = \hat{h}(S).
\end{equation*}
Computing $\Sigma_{1,t}$, we have
\begin{align*}
    \Sigma_{1,t} &= \sum_{S_1\ni t}\sum_{S_2\in T_2(S_1)} \hat{F}(S_1,S_2)^2 \\
    &= \sum_{S_1\ni t}\hat{F}(S_1,\emptyset)^2 + \hat{F}(S_1,S_1)^2 \\
    &= \sum_{S_1\ni t}\hat{f}(S_1)^2 + \sum_{S_1\ni t}\hat{h}(S_1)^2 \\
    &= \infl{t}{f} + \infl{t}{h}.
\end{align*}
Similarly, we have
\begin{equation*}
    \Sigma_{2,t} = \infl{t}{g} + \infl{t}{h}.
\end{equation*}
A simple application of Cauchy-Schwarz yields
\begin{equation}
\label{eq:crnaivfin}
    \Sigma_{1,t}^2 + \Sigma_{2,t}^2 \leq 2\infl{t}{f}^2 + 2\infl{t}{g}^2
        + 4\infl{t}{h}^2.
\end{equation}
Substituting \eqref{eq:crnaivfin} into the bound of \thref{cr:naive} yields the
desired inequality.
\end{proof}
\fi

Clearly, for separable functions BVIP-1 yields a bound which is asymptotically
tighter than that of BVIP-2 by a factor of $9^k$. This is due to the fact that
$\abs{T_2(S_1)}$ and $\abs{T_1(S_2)}$ are constants for the case of separable
functions. Note that this is not a general phenomenon: we can define functions
such that $\abs{T_2(S_1)},\abs{T_1(S_2)} \geq 9^k$, in which case BVIP-2 would
provide a tighter bound.  Nonetheless, for bivariate functions in which the
interaction between inputs is not too strong or for functions of high degree,
BVIP-1 will be tighter than the naive baseline of BVIP-2.

The fact that BVIP-1 is looser than BVIP-2 for some functions is evidence that
our analysis is not perfect. It is left to future work to investigate and
quantify the effect of the maximum degree and the interaction of the two inputs
on the relative performance of these invariance principles. Furthermore, it is
possible that other methods for proving the BIP would naturally lead to other
bivariate invariance principles which may further elucidate this tradeoff or
reveal new aspects of the problem. Finally, we also note that the bivariate
method in this paper could potentially be extended to address multivariate,
multilinear polynomials.

\bibliographystyle{IEEEtran}
\bibliography{IEEEabrv,itw-22}

\end{document}